\documentclass[aps,prl,twocolumn,10pt,superscriptaddress]{revtex4-2}

\usepackage{amsmath,amssymb,amsfonts}
\usepackage{graphicx}
\usepackage{bm}
\usepackage{xcolor}
\usepackage{amsthm}

\newtheorem{proposition}{Proposition}

\usepackage{hyperref}

\begin{document}

\title{Transient Acceleration and Cross-Dissipation Interference in Fisher-Regularized Wasserstein Gradient Flows}

\author{Michael Farmer}
\author{Abhinav Kochar}
\author{Yugyung Lee}
\affiliation{University of Missouri--Kansas City}

%\date{}

\begin{abstract}
We study transient nonequilibrium dynamics in Fisher-regularized Wasserstein gradient flows and identify a sign-changing cross-dissipation mechanism generated by the coupling between transport dissipation and Fisher-information geometry. Using the Ornstein--Uhlenbeck Fokker--Planck system as an analytically tractable setting, we derive an exact reduced variance dynamics on the Gaussian manifold,
\[
\dot{u}
=
2(1-u)
+
\frac{\varepsilon}{u},
\]
where \(u(t)=\sigma^2(t)\) denotes the variance and \(\varepsilon>0\) is the Fisher regularization strength. The reduced dynamics exhibit distinct transient regimes arising from the interaction between transport relaxation and information-geometric curvature. The associated cross-dissipation term changes sign at the critical scale \(\sigma=1\), separating a cooperative acceleration regime for localized states with \(\sigma<1\) from a transient interference regime at larger variance scales. In the subcritical regime, Fisher curvature cooperates with transport dissipation and accelerates the descent of the baseline free energy, whereas beyond the critical transition, the Fisher contribution partially opposes the Ornstein--Uhlenbeck pullback and generates transient overshoot toward a displaced Fisher-regularized equilibrium.

We further establish a bounded transient-acceleration-window result, showing that the cooperative acceleration phase has finite duration with an upper bound determined solely by the Fisher regularization strength and independent of the initial condition. Finite-difference simulations are consistent with the analytical predictions and suggest that qualitatively similar sign-transition behavior may persist beyond Gaussian closure for non-Gaussian initial conditions, including bimodal and Laplace distributions. Overall, the results provide a transient dynamical perspective on Fisher-regularized dissipative systems and show how information-geometric curvature can reorganize intermediate-time Wasserstein relaxation through cooperative acceleration, transient interference, overshoot, and equilibrium displacement while preserving the globally dissipative structure of the flow.
\end{abstract}

\maketitle

%% main text
\section{Introduction}

Dissipative gradient flows provide a fundamental mathematical framework for describing relaxation phenomena in statistical physics, nonlinear diffusion, kinetic theory, and optimal transport \cite{jordan1998variational,otto2001geometry,ambrosio2008gradient,carrillo2003kinetic,arnold2001convex}. In particular, Wasserstein gradient flows characterize the evolution of probability densities through transport-driven dissipation of free-energy functionals and establish a variational formulation for irreversible dynamics and Fokker--Planck equations \cite{villani2003topics,villani2009optimal}. These formulations reveal how transport geometry governs convergence toward equilibrium through monotone free-energy dissipation and minimizing-movement dynamics.

Many dissipative systems additionally involve geometric or information-theoretic regularization mechanisms that modify transient relaxation behavior. Among these, Fisher-information regularization arises naturally in information geometry, entropy-production theory, quantum drift-diffusion systems, Schrödinger-type formulations, and higher-order variational models \cite{arnold2001convex, amari2000methods,carlen1991superadditivity,gianazza2009wasserstein,dolbeault2006logarithmic,jungel2009transport,matthes2009family}. Related structures also appear through the Bohm quantum potential and Fisher--Rao geometric formulations of density evolution \cite{jungel2009transport,bohm1952suggested,madelung1927quantentheorie,jungel2001quasi}. Such regularization introduces curvature-dependent effects that become particularly important in strongly localized regimes, where higher-order geometric contributions can substantially influence dissipative dynamics.

While Fisher-information regularization has been studied extensively in gradient-flow theory, entropy dissipation, quantum diffusion, and higher-order Wasserstein dynamics, most existing analyses focus primarily on regularity, asymptotic convergence, equilibrium structure, or well-posedness. By contrast, considerably less attention has been devoted to transient nonequilibrium dynamics generated by the coupling between transport dissipation and Fisher curvature. In particular, the intermediate-time consequences of this interaction in strongly localized regimes remain insufficiently understood from a nonlinear-dynamical perspective.

In this work, we study Fisher-regularized Wasserstein gradient flows and identify a sign-changing cross-dissipation mechanism that produces transient acceleration and interference across distinct variance regimes. Rather than altering asymptotic stability itself, the Fisher contribution reorganizes the intermediate-time relaxation structure of the flow through cooperative acceleration, transient interference, overshoot, and equilibrium displacement. Our analysis focuses on the Ornstein--Uhlenbeck Fokker--Planck system, which provides an analytically tractable setting for isolating the interaction between transport dissipation and Fisher-information curvature.

The coupling is defined directly at the PDE level through the cross-dissipation interaction
\[
D_\times(t)
=
-
\int
\rho \nabla \mu_0 \cdot \nabla \mu_F \, dx,
\]
which measures the interaction between the baseline transport dissipation and the Fisher-curvature contribution. Under this convention, positive values of $D_\times(t)$ correspond to transient interference, whereas negative values correspond to cooperative acceleration. Similar dissipation decompositions arise in entropy-production and higher-order diffusion analyses; however, their role as sign-changing transient mechanisms in Fisher-regularized Wasserstein systems remains comparatively less understood \cite{arnold2001convex,gianazza2009wasserstein,jungel2009transport,matthes2009family}.

To explicitly analyze this interaction, we restrict the dynamics to the Gaussian manifold. Within the Gaussian ansatz, the Fisher-regularized dynamics reduce exactly to the variance equation
\[
\dot{u}
=
2(1-u)
+
\frac{\varepsilon}{u},
\]
where
\[
u(t)\equiv\sigma^2(t)
\]
denotes the evolving variance and $\varepsilon>0$ controls the Fisher regularization strength. These reduced dynamics are not introduced merely as a low-dimensional approximation; rather, they provide an analytically explicit realization of the PDE-level interaction between transport relaxation and Fisher curvature.

Within this reduced setting, the cross-dissipation interaction admits the closed-form expression
\[
D_\times(\sigma)
=
\frac{\sigma^2-1}{2\sigma^4}.
\]
This coefficient changes sign at the critical scale $\sigma=1$, separating a cooperative acceleration regime for localized states with $\sigma<1$ from a transient interference regime for larger variance scales. In the subcritical regime, Fisher curvature cooperates with transport dissipation and accelerates the descent of the baseline free energy. Beyond the critical transition, the Fisher contribution partially opposes the Ornstein--Uhlenbeck pullback, producing transient overshoot before relaxation toward a displaced Fisher-regularized equilibrium.

Prior work on Fisher-information regularization, quantum drift-diffusion, DLSS-type equations, and higher-order Wasserstein gradient flows has primarily emphasized regularity, entropy dissipation, asymptotic convergence, equilibrium structure, and well-posedness \cite{arnold2001convex,gianazza2009wasserstein,dolbeault2006logarithmic,jungel2009transport,matthes2009family,fischer2014infinite,bukal2021well}. The present work instead isolates the cross-dissipation interaction between Ornstein--Uhlenbeck transport relaxation and Fisher-information curvature as a sign-changing transient mechanism. This perspective connects the critical sign transition of the interaction coefficient to cooperative acceleration, transient interference, overshoot, and bounded acceleration windows, thereby highlighting an intermediate-time dynamical effect that is not visible from asymptotic dissipation alone.

We further establish a bounded transient-acceleration-window result showing that the cooperative acceleration phase has finite duration, with an upper bound determined solely by the Fisher regularization strength and independent of the initial condition. Beyond the Gaussian manifold, finite-difference simulations with bimodal and Laplace initial conditions exhibit qualitatively similar transient behavior, including cooperative acceleration, sign-transition dynamics, and transient overshoot. These observations suggest that related transient interaction mechanisms may extend beyond Gaussian-manifold dynamics, although the present analysis does not constitute a general PDE-level characterization beyond the reduced setting.

\begin{figure*}[t]
    \centering
    \includegraphics[width=\textwidth]{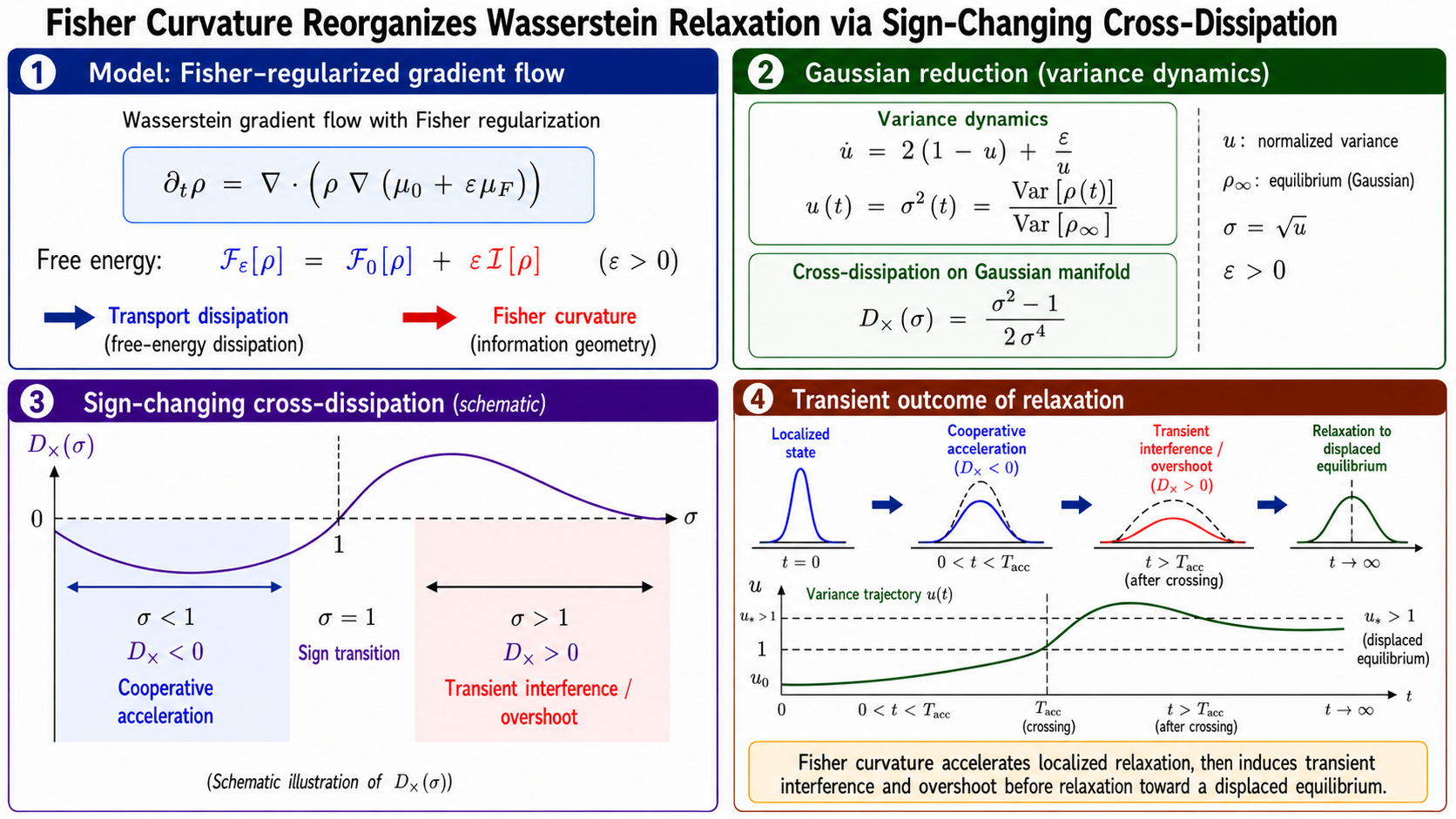}
    \caption{
    Graphical overview of the proposed Fisher-regularized Wasserstein relaxation framework.
    The model combines transport dissipation with Fisher-information curvature, reduces the
    dynamics to a Gaussian variance equation, and reveals a sign-changing cross-dissipation term that explains cooperative acceleration for localized states and transient interference
    or overshoot before relaxation to a displaced equilibrium.
    }
    \label{fig:graphical_abstract}
\end{figure*}

Figure~\ref{fig:graphical_abstract} summarizes the main mechanism studied in this work.
Starting from a Fisher-regularized Wasserstein gradient flow, the dynamics reduce on the
Gaussian manifold to a scalar variance equation. This reduction exposes a cross-dissipation
term whose sign changes at the equilibrium variance, explaining why Fisher curvature can
initially accelerate localized relaxation but later induce transient interference and overshoot
before convergence to a displaced equilibrium.

The main contributions of this work are summarized as follows:
\begin{itemize}
\item We identify a sign-changing cross-dissipation interaction governing the coupling between transport dissipation and Fisher-information curvature.

\item We derive an analytically tractable Gaussian-manifold reduction that exhibits cooperative acceleration, transient interference, overshoot, and equilibrium displacement.

\item We establish a bounded transient-acceleration-window result whose upper bound depends only on the Fisher regularization strength.

\item We provide numerical evidence suggesting that qualitatively similar sign-transition behavior may persist beyond Gaussian closure for representative non-Gaussian initial conditions.
\end{itemize}

The present work is intended primarily as a transient dynamical study of Fisher-regularized Wasserstein flows rather than a complete classification theory for information-regularized dissipative systems. The analysis is developed in a one-dimensional setting, combining analytically tractable Gaussian-manifold reductions with representative numerical experiments beyond Gaussian closure.

The remainder of the paper is organized as follows. Section~2 introduces Fisher-regularized Wasserstein gradient flows and the associated dissipation structure. Section~3 presents analytical and numerical investigations of the transient interaction dynamics for Gaussian and non-Gaussian initial conditions. Section~4 discusses broader implications and limitations of the transient interaction framework. Finally, Section~5 concludes the paper.

\section{Fisher-Regularized Wasserstein Gradient Flows}
\label{sec:background}

\subsection{Gradient Flow Formulation}

Let $\rho(x,t)$ denote a probability density on $\mathbb{R}^d$ satisfying
\[
\int_{\mathbb{R}^d}\rho(x,t)\,dx=1,
\]
and consider the baseline free-energy functional
\[
F_0[\rho]
=
\int_{\mathbb{R}^d}
\left(
V(x)\rho(x)
+
\rho(x)\log\rho(x)
\right)\,dx,
\]
where $V(x)$ is a confining potential. The associated Wasserstein gradient flow generates the classical Fokker--Planck equation
\cite{jordan1998variational,otto2001geometry,ambrosio2008gradient,carrillo2003kinetic}
\[
\partial_t\rho
=
\nabla\cdot
\left(
\rho\nabla\mu_0
\right),
\]
where the baseline chemical potential is
\[
\mu_0
=
\frac{\delta F_0}{\delta\rho}
=
V(x)+\log\rho .
\]
Along this flow, the baseline free energy satisfies the dissipation identity
\[
\frac{dF_0}{dt}
=
-
\int_{\mathbb{R}^d}
\rho |\nabla\mu_0|^2\,dx
\leq 0.
\]
This identity reflects the entropy-dissipative structure of the Wasserstein flow and implies monotone relaxation toward equilibrium
\cite{arnold2001convex,gianazza2009wasserstein}.

To incorporate information-geometric regularization, we augment the baseline free energy with the Fisher information functional
\[
\mathcal{I}[\rho]
=
\int_{\mathbb{R}^d}
|\nabla\sqrt{\rho}|^2\,dx .
\]
The Fisher-regularized free energy is defined by
\[
F_\varepsilon[\rho]
=
F_0[\rho]
+
\varepsilon \mathcal{I}[\rho],
\]
where $\varepsilon>0$ controls the regularization strength. Its variational derivative is
\[
\mu_\varepsilon
=
\frac{\delta F_\varepsilon}{\delta\rho}
=
\mu_0+\varepsilon\mu_F,
\]
where
\[
\mu_F
=
-
\frac{\Delta\sqrt{\rho}}{\sqrt{\rho}} .
\]
The Fisher contribution is closely related to the Bohm quantum potential, Fisher--Rao geometry, and quantum drift-diffusion formulations
\cite{jungel2009transport, bohm1952suggested,madelung1927quantentheorie,jungel2001quasi}.

The Fisher-regularized Wasserstein gradient flow therefore satisfies
\[
\partial_t\rho
=
\nabla\cdot
\left[
\rho\nabla
\left(
\mu_0+\varepsilon\mu_F
\right)
\right].
\]
Related higher-order Wasserstein, Fisher-information, and quantum diffusion structures have been studied in nonlinear fourth-order dissipative equations, DLSS-type models, and transport-driven variational flows
\cite{otto2001geometry,gianazza2009wasserstein,dolbeault2006logarithmic,jungel2009transport, matthes2009family,fischer2014infinite}.
Unlike purely transport-driven dissipation, the Fisher contribution introduces curvature-dependent higher-order effects that become especially pronounced for localized states and short-scale dynamics.

\subsection{Cross-Dissipation Interaction}

We next examine the dissipation of the baseline free energy $F_0[\rho]$
along the Fisher-regularized dynamics. Differentiating $F_0[\rho]$ along the
flow gives
\[
\frac{dF_0}{dt}
=
-
\int
\rho |\nabla\mu_0|^2\,dx
-
\varepsilon
\int
\rho
\nabla\mu_0
\cdot
\nabla\mu_F
\,dx .
\]
The first term is the classical transport dissipation, while the second term
captures the interaction between transport relaxation and Fisher-information
curvature.

We define the cross-dissipation interaction by
\[
D_\times(t)
=
-
\int
\rho
\nabla\mu_0
\cdot
\nabla\mu_F
\,dx .
\]
With this convention, the dissipation identity becomes
\[
\frac{dF_0}{dt}
=
-
\int
\rho |\nabla\mu_0|^2\,dx
+
\varepsilon D_\times(t).
\]
Thus, positive values of $D_\times(t)$ partially offset the baseline transport
dissipation and correspond to transient interference. Conversely, negative
values of $D_\times(t)$ make the right-hand side more negative and therefore
correspond to cooperative acceleration of the descent of $F_0$.

Importantly, $D_\times(t)$ is defined at the PDE level and does not rely on
the Gaussian-manifold reduction introduced below. The Gaussian-manifold
analysis provides an analytically explicit realization of this interaction and
allows the associated sign transition to be characterized in closed form
within the reduced setting.

In the Ornstein--Uhlenbeck setting considered below, the interaction admits
the Gaussian-manifold representation
\[
D_\times(\sigma)
=
\frac{\sigma^2-1}{2\sigma^4},
\]
which changes sign at the critical scale $\sigma=1$.

This sign convention can also be verified directly from the reduced Gaussian
variance dynamics. For the centered Gaussian family, the baseline free energy
satisfies
\[
F_0'(u)
=
\frac{u-1}{2u},
\]
and the Fisher-regularized variance equation is
\[
\dot u
=
2(1-u)+\frac{\varepsilon}{u}.
\]
Therefore,
\[
\frac{dF_0}{dt}
=
F_0'(u)\dot u
=
-\frac{(u-1)^2}{u}
+
\varepsilon \frac{u-1}{2u^2}.
\]
Since
\[
\int \rho |\nabla\mu_0|^2\,dx
=
\frac{(u-1)^2}{u},
\]
this agrees with
\[
\frac{dF_0}{dt}
=
-
\int
\rho |\nabla\mu_0|^2\,dx
+
\varepsilon D_\times(t),
\]
and confirms that, on the Gaussian manifold,
\[
D_\times(\sigma)
=
\frac{\sigma^2-1}{2\sigma^4}.
\]

\begin{proposition}[Sign transition of the cross-dissipation interaction]
For Gaussian states with variance
\[
u(t)\equiv\sigma^2(t),
\]
the cross-dissipation interaction
\[
D_\times(\sigma)
=
\frac{\sigma^2-1}{2\sigma^4}
\]
has a unique sign transition at the critical scale
\[
\sigma=1.
\]
Moreover,
\[
D_\times(\sigma)<0
\quad
(\sigma<1),
\]
and
\[
D_\times(\sigma)>0
\quad
(\sigma>1).
\]
Consequently, Fisher curvature cooperates with transport dissipation below
the critical scale and partially opposes the transport-driven pullback above
the critical transition.
\end{proposition}

\begin{proof}
Since
\[
D_\times(\sigma)
=
\frac{\sigma^2-1}{2\sigma^4},
\]
the denominator is strictly positive for $\sigma>0$. Therefore, the sign of
$D_\times(\sigma)$ is determined entirely by the numerator $\sigma^2-1$,
which changes sign uniquely at $\sigma=1$. Hence,
\[
D_\times(\sigma)<0
\quad
(\sigma<1),
\]
and
\[
D_\times(\sigma)>0
\quad
(\sigma>1).
\]
The interpretation follows from
\[
\frac{dF_0}{dt}
=
-
\int
\rho |\nabla\mu_0|^2\,dx
+
\varepsilon D_\times(t),
\]
since negative values of $D_\times$ strengthen the descent of $F_0$, whereas
positive values partially offset the baseline transport dissipation.
\end{proof}

Thus, on the Gaussian manifold, the sign-changing structure of
$D_\times(\sigma)$ separates the dynamics into a cooperative acceleration
regime for localized states with $\sigma<1$ and a transient interference
regime for broader states with $\sigma>1$.

\subsection{Ornstein--Uhlenbeck Setting and Gaussian Reduction}

To analyze the transient dynamics explicitly, we consider the one-dimensional Ornstein--Uhlenbeck potential
\[
V(x)=\frac{x^2}{2}.
\]
The corresponding equilibrium distribution is
\[
\rho_\ast(x)
=
\frac{1}{\sqrt{2\pi}}
\exp\left(
-\frac{x^2}{2}
\right).
\]

We restrict the dynamics to the centered Gaussian manifold
\[
\rho(x,t)
=
\frac{1}{\sqrt{2\pi u(t)}}
\exp\left(
-\frac{x^2}{2u(t)}
\right),
\]
where $u(t)=\sigma^2(t)$ denotes the variance. For the baseline Ornstein--Uhlenbeck Wasserstein flow, the variance evolves according to
\[
\dot{u}_0=2(1-u).
\]
For a centered Gaussian density,
\[
\mathcal{I}[\rho]
=
\int |\nabla\sqrt{\rho}|^2\,dx
=
\frac{1}{4u}.
\]
Projecting the Fisher contribution onto the Gaussian variance coordinate gives the additional variance drift
\[
\dot{u}_F=\frac{1}{u}.
\]
Therefore, under Fisher regularization with strength $\varepsilon>0$, the Gaussian-manifold variance dynamics reduce to
\[
\dot{u}
=
2(1-u)
+
\frac{\varepsilon}{u}.
\]

This reduced equation provides an analytically explicit realization of the interaction between transport dissipation and Fisher curvature. It reveals three qualitatively distinct transient regimes:
(i) a cooperative acceleration regime for states with $\sigma<1$,
(ii) an interference regime beyond the critical transition, and
(iii) convergence toward a displaced Fisher-regularized equilibrium.

The displaced equilibrium is obtained by solving
\[
2(1-u)+\frac{\varepsilon}{u}=0,
\]
which gives
\[
u_\ast
=
\frac{1+\sqrt{1+2\varepsilon}}{2}
>
1.
\]

\begin{proposition}[Bounded acceleration-window crossing time]
For Gaussian initial data with variance
\[
u_0<1,
\]
the transition time required to reach the critical scale
\[
u=1
\]
is given by
\[
T_{\mathrm{acc}}(u_0)
=
\int_{u_0}^{1}
\frac{u\,du}
{2u(1-u)+\varepsilon}.
\]
Moreover, for fixed $\varepsilon>0$, the transition time is finite and monotonically decreasing as the initial variance $u_0$ increases. In particular,
\[
T_{\mathrm{acc}}(u_0)
\le
\int_{0}^{1}
\frac{u\,du}
{2u(1-u)+\varepsilon}
=: C_\varepsilon < \infty.
\]
Hence, the duration of the cooperative acceleration regime is uniformly bounded by a constant depending only on the Fisher regularization strength $\varepsilon$.
\end{proposition}

\begin{proof}
The cooperative acceleration regime terminates when the trajectory reaches the critical scale $u=1$. From the reduced Gaussian-manifold dynamics,
\[
\dot u
=
2(1-u)
+
\frac{\varepsilon}{u},
\]
we obtain
\[
dt
=
\frac{u\,du}
{2u(1-u)+\varepsilon}.
\]
Integrating from $u_0$ to $1$ gives
\[
T_{\mathrm{acc}}(u_0)
=
\int_{u_0}^{1}
\frac{u\,du}
{2u(1-u)+\varepsilon}.
\]
For $\varepsilon>0$, the denominator satisfies
\[
2u(1-u)+\varepsilon>0
\]
on $0<u\leq 1$, so the integral is finite. Furthermore,
\[
T_{\mathrm{acc}}(u_0)
\le
\int_{0}^{1}
\frac{u\,du}
{2u(1-u)+\varepsilon}
=
C_\varepsilon,
\]
which establishes a uniform upper bound independent of the initial variance. Monotonicity with respect to $u_0$ follows directly from the integral representation.
\end{proof}

%%%%%%%%%%%%%%%%%%%%%%%%%%%%%

\section{Numerical Experiments}
\label{sec:numerics}

\subsection{Simulation Setup}

To investigate transient dynamics beyond the analytical Gaussian-manifold reduction, we numerically solve the one-dimensional Fisher-regularized Fokker--Planck equation
\[
\partial_t\rho
=
\nabla\cdot
\left[
\rho
\nabla
\left(
V(x)+\log\rho+\varepsilon\mu_F
\right)
\right],
\]
where
\[
\mu_F
=
-
\frac{\Delta\sqrt{\rho}}{\sqrt{\rho}},
\]
under the Ornstein--Uhlenbeck potential
\[
V(x)=\frac{x^2}{2}.
\]

Simulations are performed on the bounded one-dimensional domain
\[
x\in[-L,L],
\]
with no-flux boundary conditions. Spatial derivatives are approximated using centered finite differences on a uniform grid. Time integration is implemented using a semi-implicit scheme in which the Fisher contribution is treated implicitly to improve stability in strongly localized regimes, while the lower-order transport contribution is advanced consistently with the finite-difference discretization. Related structure-preserving and variational discretizations have been studied for nonlinear fourth-order and DLSS-type Wasserstein gradient flows \cite{bukal2021well,during2010gradient}.

Unless otherwise specified, simulations use
\[
L=10,
\qquad
\Delta x = 10^{-2},
\qquad
\Delta t = 10^{-4},
\]
with Fisher regularization strength
\[
\varepsilon=10^{-4}.
\]
This small-regularization regime was selected to isolate transient curvature effects while preserving the dominant transport-driven structure of the underlying Ornstein--Uhlenbeck dynamics. Additional experiments with nearby regularization strengths produced qualitatively similar transient interaction behavior.

In the discrete implementation, the Fisher potential is evaluated as
\[
\mu_{F,h}
=
-
\frac{\Delta_h \sqrt{\rho_h+\delta}}
{\sqrt{\rho_h+\delta}},
\]
where $\Delta_h$ denotes the centered second-difference operator and $\delta>0$ is a small numerical floor used to avoid division by zero in low-density tails. After each time step, the density is renormalized to unit mass to correct small numerical quadrature errors; the mass correction remained negligible in all reported simulations. Positivity is monitored throughout the simulation, and the reported experiments use parameter settings for which the density remains nonnegative up to the numerical floor.

To assess numerical robustness, additional grid- and timestep-refinement experiments were performed. The resulting trajectories exhibited qualitatively consistent crossing behavior, sign-transition structure, overshoot behavior, and asymptotic relaxation.

Both Gaussian and non-Gaussian initial conditions are considered to evaluate the robustness of the transient interaction mechanism beyond the Gaussian-manifold setting.

\subsection{Numerical Methodology}

The acceleration-window crossing time
\[
T_{\mathrm{acc}}
\]
is defined as the first crossing time satisfying
\[
u(t)=1,
\]
where
\[
u(t)\equiv\sigma^2(t)
\]
denotes the evolving variance. For Gaussian initial conditions, the variance follows directly from the reduced Gaussian-manifold dynamics. For non-Gaussian states, the variance is evaluated numerically from the evolving density distribution,
\[
u(t)
=
\int x^2\rho(x,t)\,dx.
\]

The cross-dissipation interaction is computed using
\[
D_\times(t)
=
-
\int
\rho
\nabla\mu_0
\cdot
\nabla\mu_F
\,dx,
\]
where
\[
\mu_0=V(x)+\log\rho.
\]
In the numerical implementation, $\mu_0$, $\mu_F$, and their spatial derivatives are evaluated on the computational grid using centered finite differences. The integral defining $D_\times(t)$ is then approximated by the finite-difference quadrature
\[
D_{\times,h}(t)
=
-
\sum_j
\rho_j
\left(D_h\mu_{0,j}\right)
\left(D_h\mu_{F,j}\right)
\Delta x,
\]
where $D_h$ denotes the centered first-difference operator. The same numerical density is used to compute the variance, the crossing time, and the cross-dissipation interaction.

For all reported experiments, mass conservation, positivity, crossing times, and the sign of $D_\times(t)$ were monitored throughout the simulation. Refinement experiments with smaller $\Delta x$ and $\Delta t$ produced qualitatively consistent trajectories, confirming that the observed sign-transition and overshoot behavior are not artifacts of a particular discretization scale.

\subsection{Gaussian Initial Conditions}

We first consider centered Gaussian initial states of the form
\[
\rho_0(x)
=
\frac{1}{\sqrt{2\pi\sigma_0^2}}
\exp\left(
-\frac{x^2}{2\sigma_0^2}
\right),
\]
with varying initial variance $\sigma_0^2$.

Figure~\ref{fig:gaussian_dynamics} illustrates the variance evolution under Fisher-regularized Wasserstein dynamics for several localized initial conditions. Consistent with the reduced Gaussian-manifold dynamics, highly localized states initially exhibit rapid curvature-driven expansion induced by Fisher regularization.

In the subcritical regime
\[
\sigma<1,
\]
the cross-dissipation interaction satisfies
\[
D_\times<0,
\]
indicating cooperative behavior between Fisher curvature and transport dissipation. Consequently, the trajectories exhibit accelerated variance expansion relative to the classical Ornstein--Uhlenbeck dynamics without Fisher regularization.

The trajectories subsequently reach the critical transition scale
\[
\sigma=1,
\]
where the cross-dissipation interaction changes sign. Beyond this transition,
\[
D_\times>0,
\]
and the dynamics enter a transient interference regime characterized by overshoot relative to the displaced Fisher-regularized equilibrium.

After the overshoot phase, the trajectories gradually relax toward the displaced equilibrium
\[
u_\ast
=
\frac{1+\sqrt{1+2\varepsilon}}{2}
>
1,
\]
while preserving the globally dissipative structure of the flow.

The observed transient dynamics are qualitatively consistent with the sign-transition behavior predicted by the interaction coefficient
\[
D_\times(\sigma)
=
\frac{\sigma^2-1}{2\sigma^4},
\]
which is exact under the Gaussian-manifold reduction.

\begin{figure}[t]
\centering
\includegraphics[width=1\linewidth]{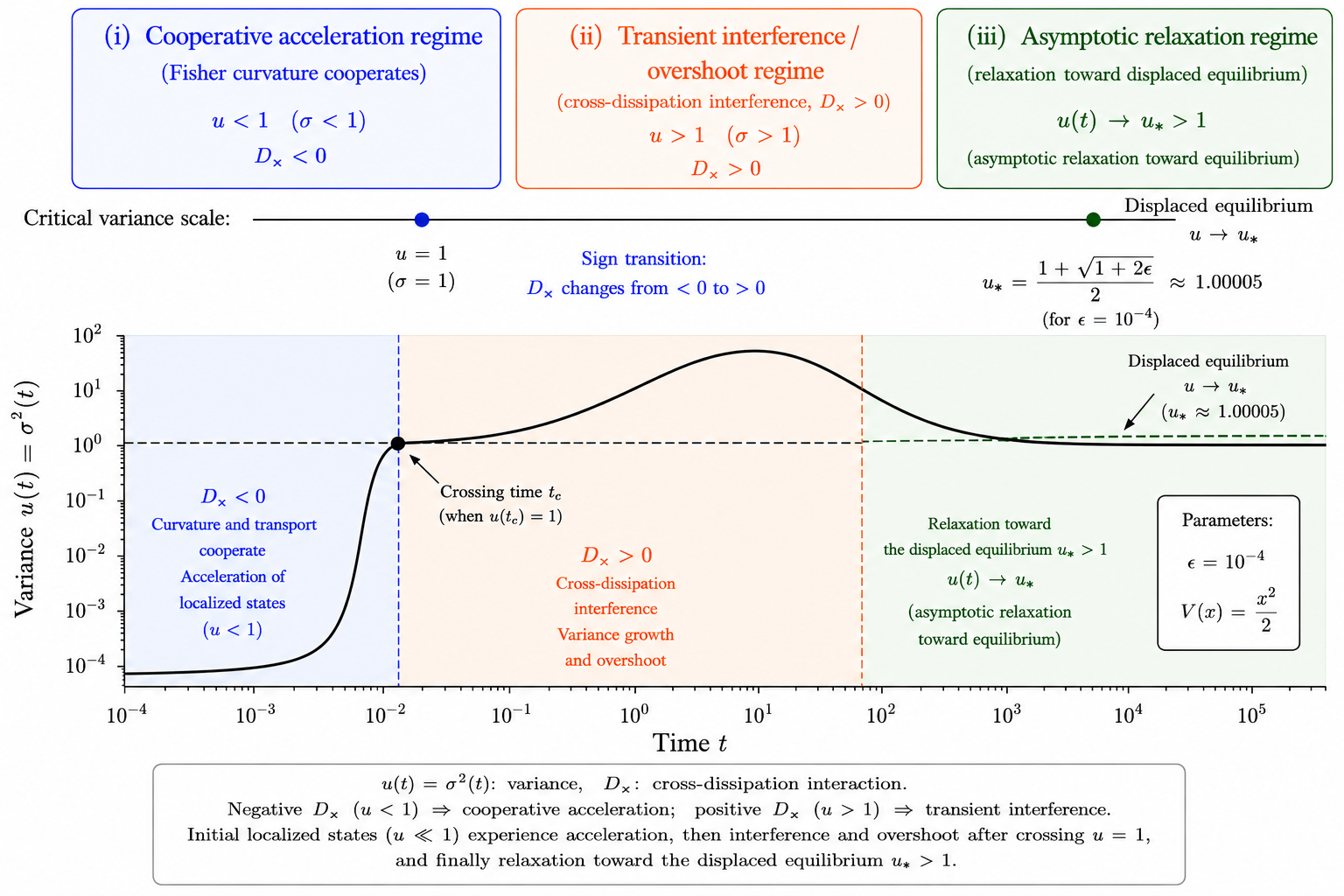}
\caption{
Variance evolution for localized Gaussian initial conditions under Fisher-regularized Wasserstein dynamics. The trajectories exhibit:
(i) cooperative acceleration for subcritical states with $\sigma<1$ and $D_\times<0$,
(ii) transient cross-dissipation interference and overshoot after crossing the critical transition $\sigma=1$ where $D_\times>0$,
and
(iii) asymptotic relaxation toward the displaced Fisher-regularized equilibrium $u_\ast>1$.
The observed sign-transition behavior is qualitatively consistent with the Gaussian-manifold interaction coefficient 
$D_\times(\sigma)
=
\frac{\sigma^2-1}{2\sigma^4}$.
}
\label{fig:gaussian_dynamics}
\end{figure}

\subsection{Cross-Dissipation Sign Transition}

To visualize the transient interaction structure directly, we compute the cross-dissipation interaction
\[
D_\times(t)
=
-
\int
\rho
\nabla\mu_0
\cdot
\nabla\mu_F
\,dx.
\]
Figure~\ref{fig:cross_dissipation} illustrates the evolution of the cross-dissipation interaction during Fisher-regularized Wasserstein dynamics. The interaction is initially negative for highly localized states,
\[
u<1,
\]
indicating cooperative behavior between Fisher curvature and transport dissipation.

As the variance approaches the critical transition scale
\[
u=1
\qquad
(\sigma=1),
\]
the interaction crosses zero and changes sign. This transition separates the cooperative acceleration regime from the transient interference regime.

In the subcritical regime,
\[
D_\times(t)<0,
\]
Fisher curvature cooperates with transport-driven dissipation, producing cooperative acceleration and rapid variance expansion relative to the classical Ornstein--Uhlenbeck dynamics.

After the crossing transition,
\[
D_\times(t)>0,
\]
the dynamics enter a transient interference regime in which the cross-dissipation interaction partially offsets the transport-driven dissipation structure. This interference produces overshoot relative to the displaced Fisher-regularized equilibrium.

As the dynamics approach the asymptotic regime,
\[
u(t)\to u_\ast,
\]
the positive interaction gradually weakens and decays toward zero from above, while the solution relaxes toward the displaced equilibrium state. The observed sign-transition structure is qualitatively consistent with the analytical Gaussian-manifold prediction
\[
D_\times(\sigma)
=
\frac{\sigma^2-1}{2\sigma^4}.
\]

\begin{figure}[t]
\centering
\includegraphics[width=1\linewidth]{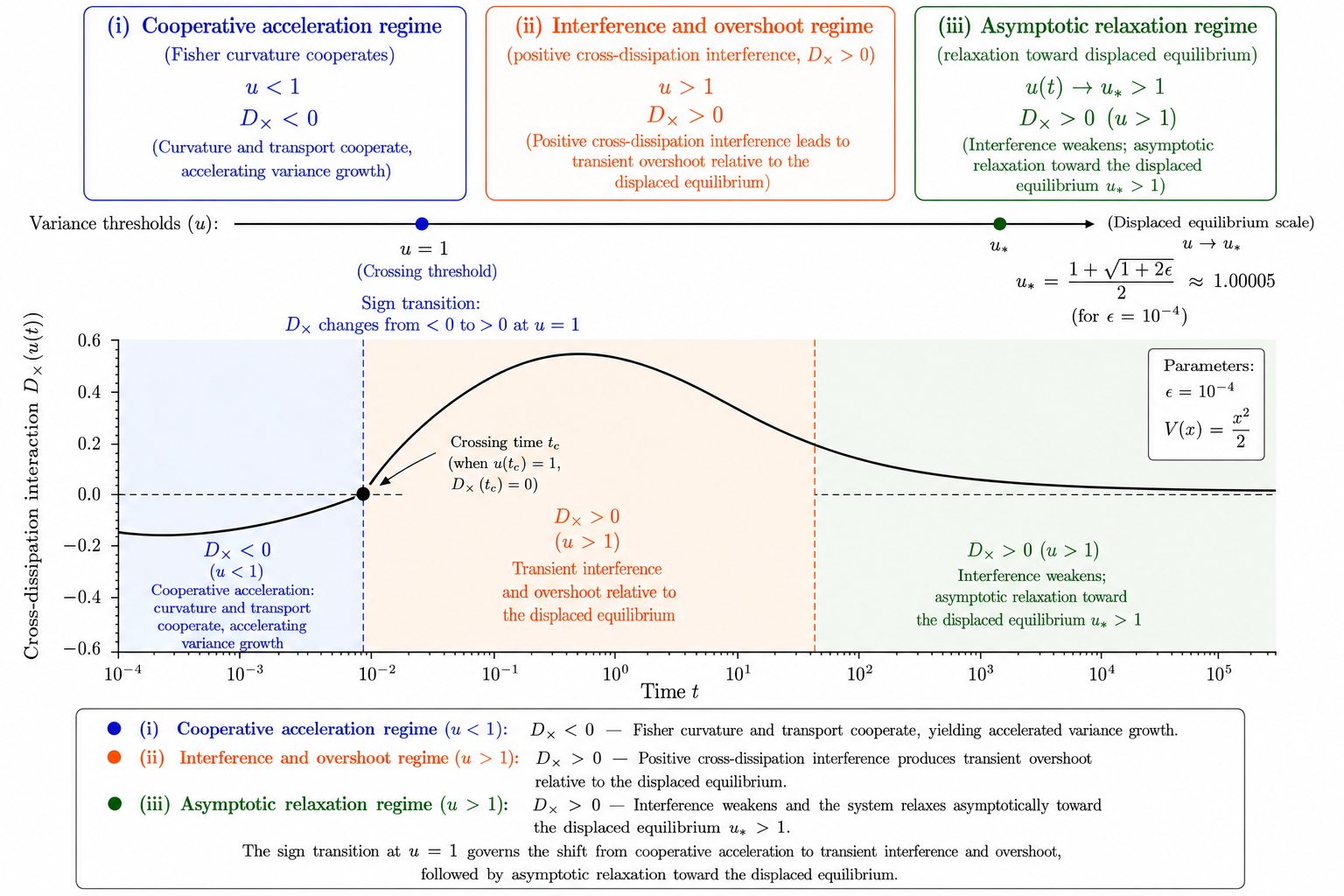}
\caption{
Evolution of the cross-dissipation interaction $D_\times(u(t))$ under Fisher-regularized Wasserstein dynamics. For localized states with $u<1$, the interaction remains negative and produces cooperative acceleration between Fisher curvature and transport-driven dissipation. The interaction crosses zero at the critical transition $u=1$, after which positive cross-dissipation interference generates overshoot relative to the displaced Fisher-regularized equilibrium. As the dynamics approach equilibrium, the positive interaction weakens asymptotically and decays toward zero from above.
}
\label{fig:cross_dissipation}
\end{figure}

\subsection{Non-Gaussian Initial Conditions}

To investigate whether qualitatively similar transient interaction behavior may extend beyond Gaussian-manifold dynamics, we additionally consider non-Gaussian initial states, including bimodal and Laplace distributions.

The bimodal initial condition is defined as
\[
\rho_0(x)
=
\frac12
\mathcal{N}(-a,\sigma_0^2)
+
\frac12
\mathcal{N}(a,\sigma_0^2),
\]
while the Laplace initial condition is given by
\[
\rho_0(x)
=
\frac{1}{2b}
\exp\left(
-\frac{|x|}{b}
\right).
\]

To further assess robustness, we additionally consider asymmetric bimodal states and unequal-scale Laplace configurations. These initializations introduce heterogeneous concentration scales and localized geometric structure beyond the Gaussian-manifold setting.

Figure~\ref{fig:nongaussian} presents representative transient dynamics for multiple non-Gaussian initialization families. Despite quantitative differences across distributions, the observed transient dynamics remain qualitatively consistent with the Gaussian-manifold predictions.

Across all initialization families, strongly localized states initially exhibit cooperative acceleration associated with the subcritical regime
\[
u<1,
\qquad
D_\times<0,
\]
where Fisher curvature cooperates with transport-driven dissipation and accelerates variance growth relative to the corresponding Ornstein--Uhlenbeck dynamics without Fisher regularization.

As the trajectories approach the critical transition scale
\[
u=1,
\]
the interaction changes sign and the dynamics enter a transient interference regime characterized by
\[
D_\times>0.
\]
This post-crossing regime produces overshoot relative to the displaced Fisher-regularized equilibrium before eventual asymptotic relaxation toward
\[
u_\ast>1.
\]

Although the detailed trajectories vary across initialization families, the qualitative sequence of:
(i) cooperative acceleration,
(ii) sign transition,
(iii) transient interference and overshoot,
and
(iv) asymptotic relaxation
is consistently observed across bimodal, asymmetric bimodal, and Laplace initializations.

Asymmetric and unequal-scale configurations additionally produce shifted crossing times and modified transient trajectories, indicating that the transient dynamics depend on geometric structure beyond variance alone.

These numerical results suggest that qualitatively similar transient interaction mechanisms may extend beyond Gaussian closure, although the present experiments do not constitute a general PDE-level characterization beyond the reduced Gaussian-manifold setting.

To further examine this behavior quantitatively, we compare the acceleration-window crossing time
\[
T_{\mathrm{acc}}
\]
against the initial information distance
\[
D_{\mathrm{KL}}
(\rho_0\Vert\rho_\ast)
\]
across multiple initialization families.

Figure~\ref{fig:delay_scaling} shows that all initialization families exhibit increasing acceleration-window crossing times with increasing Kullback--Leibler divergence from equilibrium. Although the quantitative trends differ across initialization classes, the overall relationship remains monotonically increasing across geometrically distinct states.

In particular, Gaussian initializations consistently produce the shortest crossing times, while increasingly non-Gaussian and heavy-tailed configurations exhibit larger transient interaction windows. These observations support the interpretation that larger information distance from equilibrium is associated with longer transient interaction dynamics, while avoiding claims of a universal scaling law.

\begin{figure}[ht!]
\centering
\includegraphics[width=0.9\linewidth]{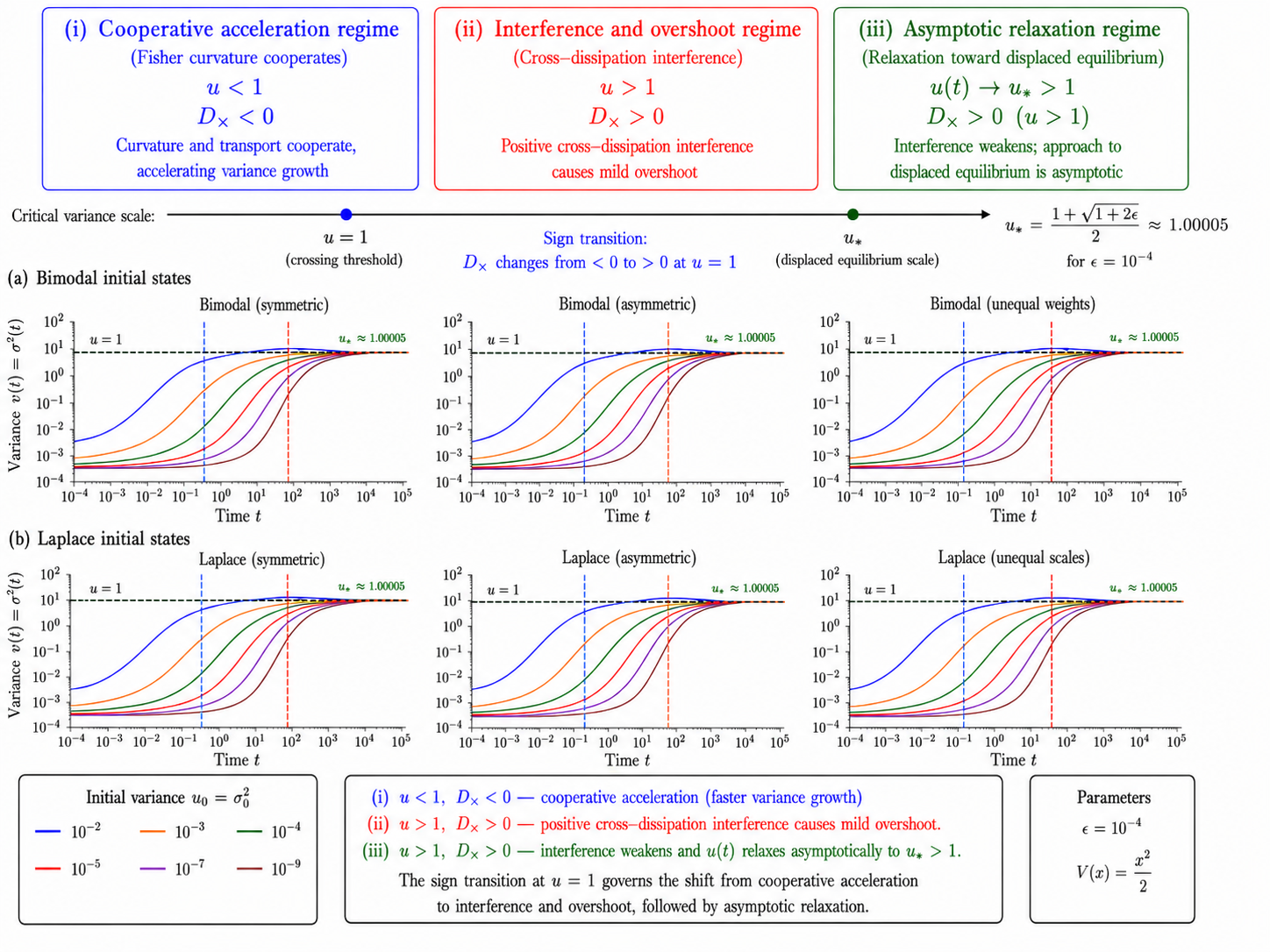}
\caption{
Representative transient dynamics for non-Gaussian initial conditions. Bimodal and Laplace initial states exhibit:
(i) cooperative acceleration in the subcritical regime $u<1$ with $D_\times<0$,
(ii) a sign transition near the critical scale $u=1$,
(iii) transient cross-dissipation interference and overshoot in the post-crossing regime with $D_\times>0$,
and
(iv) asymptotic relaxation toward the displaced Fisher-regularized equilibrium $u_\ast>1$.
Asymmetric and unequal-scale configurations produce shifted acceleration-window crossing times and modified transient trajectories. The observed transient behavior is qualitatively consistent with the Gaussian-manifold interaction dynamics.
}
\label{fig:nongaussian}
\end{figure}

\begin{figure}[ht!]
\centering
\includegraphics[width=0.85\linewidth]{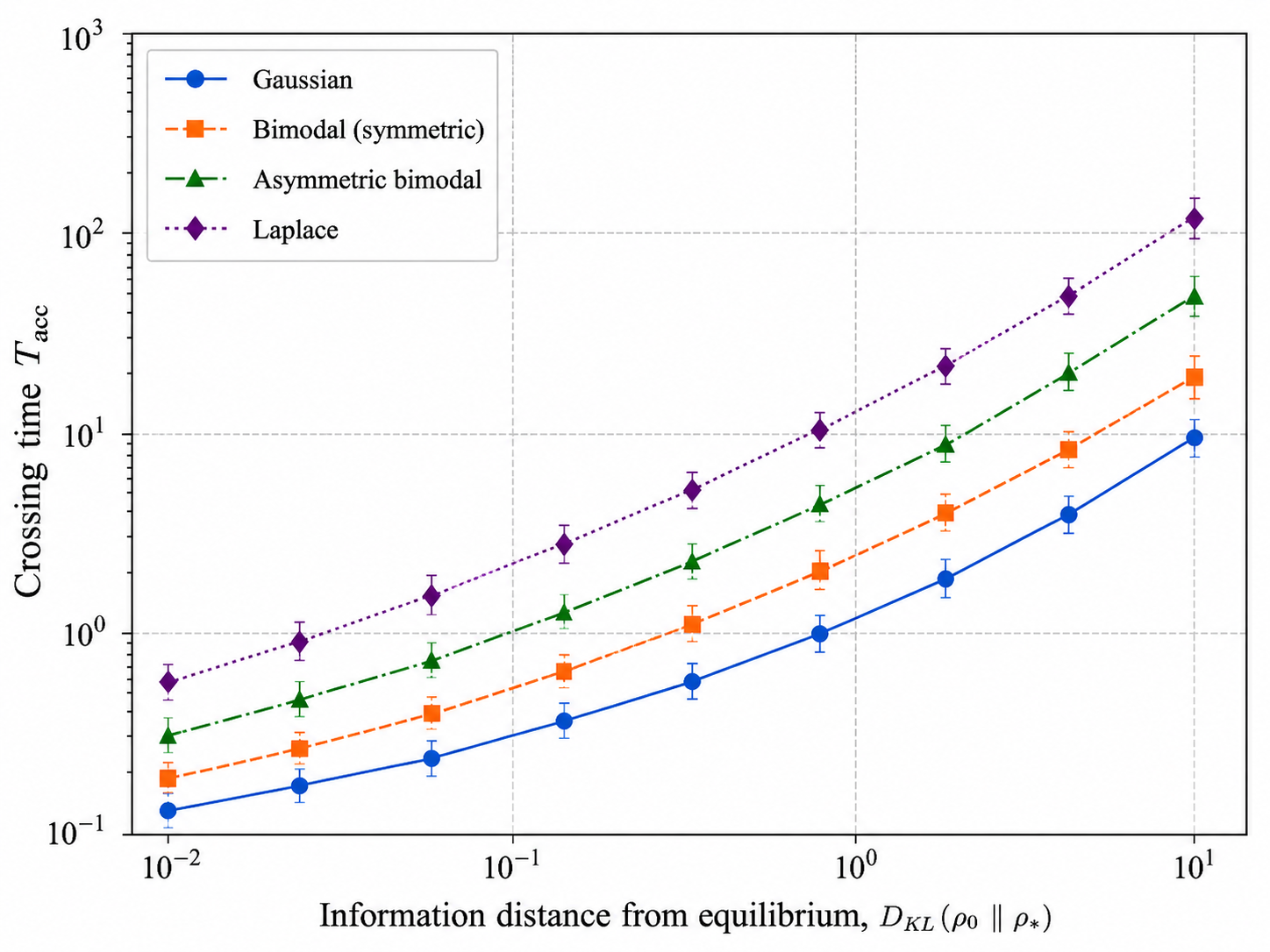}
\caption{
Representative acceleration-window crossing times $T_{\mathrm{acc}}$ versus information distance from equilibrium for Gaussian, bimodal, asymmetric bimodal, and Laplace initial conditions. All initialization families exhibit increasing crossing times with increasing Kullback--Leibler divergence from equilibrium. Although the quantitative trends differ across initialization classes, the overall relationship remains monotonically increasing across geometrically distinct non-Gaussian states.
}
\label{fig:delay_scaling}
\end{figure}
\section{Discussion}
\label{sec:discussion}

The present analysis shows that Fisher-information regularization can play a role beyond smoothing or stabilization in Wasserstein gradient flows. In the Ornstein--Uhlenbeck setting, it induces a sign-changing cross-dissipation interaction between transport relaxation and Fisher-information curvature. This interaction reorganizes intermediate-time dynamics through cooperative acceleration, transient interference, overshoot, and relaxation toward a displaced equilibrium.

On the Gaussian manifold, the cross-dissipation coefficient
\[
D_\times(\sigma)
=
\frac{\sigma^2-1}{2\sigma^4}
\]
serves as the organizing quantity for this transient behavior. The coefficient changes sign at the critical scale $\sigma=1$. For localized states with $\sigma<1$, the interaction is negative and strengthens the descent of the baseline free energy, producing a cooperative acceleration regime in which Fisher curvature acts in concert with transport dissipation. Once the trajectory crosses the critical scale, the interaction becomes positive and partially offsets the baseline transport dissipation. This produces a transient interference regime and overshoot before relaxation toward the displaced Fisher-regularized equilibrium.

This interference should not be interpreted as instability or loss of dissipation for the full Fisher-regularized flow. Instead, it reflects a redistribution of dissipation between the baseline free-energy component and the Fisher-information contribution. In this sense, the Fisher term modifies the transient dissipation geometry while preserving the global dissipative structure of the regularized dynamics.

Equivalently, the monotone Lyapunov functional for the regularized flow is
\[
F_\varepsilon[\rho]
=
F_0[\rho]
+
\varepsilon I[\rho],
\]
not the baseline functional $F_0[\rho]$ alone. Therefore, $F_0$ need not decrease monotonically along the Fisher-regularized dynamics. The observed non-monotone behavior of $F_0$, including descent past the unregularized minimum followed by relaxation toward the displaced equilibrium, is therefore a diagnostic of the equilibrium displacement induced by Fisher regularization rather than a failure of dissipation.

The numerical experiments suggest that qualitatively similar transient behavior may occur beyond the Gaussian manifold. Bimodal, asymmetric bimodal, and Laplace initial conditions exhibit the same qualitative sequence of cooperative acceleration, sign transition, transient interference, overshoot, and asymptotic relaxation. However, the detailed trajectories and crossing times depend on geometric features beyond variance alone. Thus, the Gaussian reduction captures a useful organizing mechanism, but it does not provide a complete description of the full PDE dynamics for general initial data.

The observed relationship between the acceleration-window crossing time
\[
T_{\mathrm{acc}}
\]
and the initial information distance
\[
D_{\mathrm{KL}}(\rho_0\Vert\rho_\ast)
\]
further suggests that transient interaction windows are influenced by the information-geometric distance from equilibrium. More localized or information-distant initial states tend to exhibit longer cooperative acceleration windows before entering the interference regime. Because the quantitative trends vary across initialization families, this relationship should be interpreted as qualitative rather than as a universal scaling law.

More broadly, the results illustrate how curvature-dependent regularization can reshape intermediate-time behavior in dissipative PDEs. Related transient interaction effects may arise in entropy-regularized transport, quantum drift-diffusion models, higher-order Wasserstein flows, diffusion-based generative dynamics, and other gradient-flow systems where transport dissipation is coupled to information-geometric or curvature-dependent regularization.

Several limitations remain. First, the analytical results are developed primarily in a one-dimensional Ornstein--Uhlenbeck setting and rely on a Gaussian-manifold reduction for explicit closed-form characterization. Second, the non-Gaussian experiments provide qualitative numerical evidence, but they do not establish a general PDE-level theorem beyond Gaussian closure. Third, the numerical simulations are designed to characterize transient regimes and sign-transition behavior rather than to provide an exhaustive high-precision numerical analysis of the full Fisher-regularized dynamics.

Future work should address the multidimensional structure of the Fisher cross-dissipation interaction, develop asymptotic estimates beyond Gaussian closure, and establish functional-analytic conditions under which related sign-changing transient mechanisms persist for broader classes of initial data and potentials. It would also be useful to examine whether analogous transient acceleration and interference regimes arise in other information-regularized gradient flows.

Overall, the results provide a geometric perspective on how Fisher-information curvature can reorganize nonequilibrium relaxation in Wasserstein gradient flows. The emergence of cooperative acceleration, sign-changing cross-dissipation, and transient interference highlights an intermediate-time dynamical structure that is not captured by asymptotic dissipation analysis alone.

\section{Conclusion}
\label{sec:conclusion}

In this work, we studied Fisher-regularized Wasserstein gradient flows and identified a sign-changing cross-dissipation interaction between transport dissipation and Fisher-information geometry. Using the Ornstein--Uhlenbeck Fokker--Planck system as an analytically tractable setting, we showed that Fisher regularization can reorganize intermediate-time relaxation dynamics through cooperative acceleration, transient interference, overshoot, and relaxation toward a displaced equilibrium.

The key organizing quantity is the Gaussian-manifold interaction coefficient
\[
D_\times(\sigma)
=
\frac{\sigma^2-1}{2\sigma^4},
\]
which changes sign at the critical scale $\sigma=1$. For localized states with $\sigma<1$, the interaction is negative and strengthens the descent of the baseline free energy, producing cooperative acceleration. Beyond the critical transition, the interaction becomes positive and partially offsets the baseline transport dissipation, giving rise to transient interference and overshoot before asymptotic relaxation.

We also established a bounded acceleration-window result, showing that the cooperative acceleration phase has finite duration with an upper bound determined only by the Fisher regularization strength. Numerical experiments with Gaussian and representative non-Gaussian initial conditions, including bimodal and Laplace distributions, exhibited qualitatively similar sign-transition behavior. These results suggest that the Gaussian reduction captures a useful organizing transient mechanism, while the extension to general initial data remains an open analytical problem.

More broadly, the results provide a geometric perspective on how Fisher-information curvature can reshape nonequilibrium relaxation in Wasserstein gradient flows while preserving the globally dissipative structure of the regularized dynamics. Future work will focus on multidimensional extensions, analytical characterization beyond Gaussian closure, and related transient interaction phenomena in other information-regularized gradient-flow systems.

% \makeatletter
% \def\bibsection{\section*{\refname}}
% \makeatother

\bibliographystyle{elsarticle-num}
\bibliography{references}

\end{document}